\documentclass[aps,pra,reprint,showpacs,floatfix,longbibliography]{revtex4-1}
\usepackage{graphicx, amsmath, amssymb, amsthm, braket, enumitem}
\usepackage[caption=false]{subfig}
\allowdisplaybreaks[1] 

\newtheorem{definition}{Definition}
\newtheorem{lemma}{Lemma}
\newtheorem{theorem}{Theorem}

\newcommand{\ketbra}[2]{{\left| #1 \middle\rangle \middle\langle #2 \right|}}

\newcommand{\fstate}{\phi} 
\newcommand{\ustate}{\sigma} 
\newcommand{\rstate}{\rho} 
\newcommand{\fcoef}{\alpha} 
\newcommand{\ucoef}{\beta}  
\newcommand{\rcoef}{\gamma} 

\newcommand{\avgin}[2]{\overline{#1_{#2}^{\tin}}}
\newcommand{\avgout}[2]{\overline{#1_{#2}^{\tout}}}
\newcommand{\flip}[1]{\widetilde{#1}}

\DeclareMathOperator{\tout}{out}
\DeclareMathOperator{\tin}{in}

\begin{document}

\title{Oscillatory Localization of Quantum Walks Analyzed by Classical Electric Circuits}

\author{Andris Ambainis}
	\email{andris.ambainis@lu.lv}
\author{Kri\v{s}j\={a}nis Pr\={u}sis}
	\email{krisjanis.prusis@lu.lv}
\author{Jevg\={e}nijs Vihrovs}
	\email{jevgenijs.vihrovs@lu.lv}
\author{Thomas G.~Wong}
	\email{Currently at University of Texas at Austin, twong@cs.utexas.edu}
	\affiliation{Faculty of Computing, University of Latvia, Rai\c{n}a bulv.~19, R\=\i ga, LV-1586, Latvia}

\begin{abstract}
	We examine an unexplored quantum phenomenon we call oscillatory localization, where a discrete-time quantum walk with Grover's diffusion coin jumps back and forth between two vertices. We then connect it to the power dissipation of a related electric network. Namely, we show that there are only two kinds of oscillating states, called uniform states and flip states, and that the projection of an arbitrary state onto a flip state is bounded by the power dissipation of an electric circuit. By applying this framework to states along a single edge of a graph, we show that low effective resistance implies oscillatory localization of the quantum walk. This reveals that oscillatory localization occurs on a large variety of regular graphs, including edge-transitive, expander, and high degree graphs. As a corollary, high edge-connectivity also implies localization of these states, since it is closely related to electric resistance.
\end{abstract}


\maketitle


\section{\label{sec:introduction} Introduction}

Localization was first observed in discrete-time quantum walks roughly fourteen years ago by Mackay et al.~\cite{Mackay2002}, whose numerical simulations demonstrated that an initially localized quantum walker on the two-dimensional (2D) lattice had a high probability of remaining at its initial location. This behavior was further investigated numerically by Tregenna et al.~\cite{Tregenna2003}, and it was analytically proved to persist for all time by Inui, Konishi, and Konno \cite{Inui2004}. Since this seminal work, localization in quantum walks has been an area of thriving research (see Section 2.2.9 of \cite{Venegas2012} for an overview and the references therein). Such localization is a purely quantum phenomenon, starkly different from the diffusive behavior of classical random walks. Furthermore, the ability to localize a quantum walker has potential applications in quantum optics, quantum search algorithms, and investigating topological phases \cite{Kitagawa2010,Kollar2015}.

In this paper, we introduce a new type of localization where the quantum walker jumps back and forth between two locations, so we term it \emph{oscillatory localization}. The first hint of this behavior appears in Inui, Konishi, and Konno's aforementioned analysis of the 2D walk \cite{Inui2004}, where the probability of finding the walker at its initial location is high at even times and small at odd times. Our analysis shows that this is due to the walker jumping back and forth between its initial location and an adjacent site, and in this paper, we prove that it occurs on a wide variety of graphs, including complete graphs, complete bipartite graphs, hypercubes, square lattices of high dimension, expander graphs, and high degree graphs. 

The quantum walk is defined on a graph of $N$ vertices \cite{Aharonov2001}, so the walker jumps in superposition from vertex to vertex with the edges defining the allowed transitions. For simplicity, we assume that the graph is regular with degree $d$. Then besides the position space, we include an additional $d$-dimensional internal ``coin'' degree of freedom in order to define a non-trivial walk \cite{Meyer1996a,Meyer1996b}, spanned by the $d$ directions along which the walker can hop. With this, the full Hilbert space of the system is $\mathbb{C}^N \otimes \mathbb{C}^d$. Then the quantum walk is defined by repeated applications of the operator
\begin{equation}
	\label{eq:U}
	U = S \cdot (I_N \otimes C),
\end{equation}
where $C$ is the ``coin flip'' that acts on the internal state of the system, and $S$ is the shift operator that causes the walker to move based on its internal state. As with most prior work on quantum walks and localization \cite{Kollar2015}, throughout this paper, we choose $C$ to be the ``Grover coin'' \cite{SKW2003}:
\[ C = 2 \ketbra{s_c}{s_c} - I_d, \]
where $\ket{s_c} = \sum_{i=1}^d \ket{i} / \sqrt{d}$ is the equal superposition over the coin space.

For the shift operator $S$, most papers on localization focus on 1D or 2D lattices, so they typically use the ``moving'' shift $S_m$, where the particle jumps and continues pointing in the same direction. For example, on the 1D line, $S_m (\ket{0} \otimes \ket{\rightarrow}) = \ket{1} \otimes \ket{\rightarrow}$. In our paper, however, we include non-lattice graphs as well. On non-lattice graphs, the moving shift's notion of staying in the same direction is unclear, requiring an additional labeling of directed edges that form a permutation \cite{Aharonov2001}. We mitigate this by instead using the ``flip-flop'' shift \cite{SKW2003}, where the particle jumps and then turns around. For example, on the 1D line, $S (\ket{0} \otimes \ket{\rightarrow}) = \ket{1} \otimes \ket{\leftarrow}$. Besides having a natural definition on non-lattice graphs, the flip-flop shift is also important from an algorithmic standpoint; for spatial search, where the quantum walk is supplemented by an oracle, there are cases where the flip-flop shift yields a quantum speedup while the moving shift yields no improvement over classical \cite{AKR2005}. For more benefits of the flip-flop shift, see the introduction of \cite{Wong17}.

In the next section, we give a simple example of oscillatory localization of the quantum walk on the complete graph. The analysis is straightforward enough that the exact evolution can be determined using basic linear algebra. This forms intuition for more advanced analytical techniques, beginning in Section \ref{sec:exact}. There, we determine the eigenvectors of $U^2$ with eigenvalue 1. We show that there are only two different types, which we call uniform states and flip states, and they form a complete orthogonal basis for exact oscillatory states. So the projection of an arbitrary state onto these gives a lower bound on the extent of the oscillation, as shown in Section \ref{sec:approximate}. While the projection of an arbitrary state onto uniform states is trivial, it is much more challenging to find its projection onto flip states. 

So for the rest of the paper, we develop a method for lower-bounding the projection onto flip states using classical electric networks in Section \ref{sec:circuits}. To do this, we define a bijection between flip states and circulation flows in a related graph. Since electric current is a circulation flow, we prove that oscillations on a graph occurs if the power dissipation on a related electric network is low. Then we apply this framework to certain localized starting states, showing that effective resistance can be used instead of power dissipation. That is, low electric resistance implies oscillatory localization of these states of the quantum walk. Since effective resistance is inversely related to edge-connectivity, it follows that high edge-connectivity also implies localization for the particular starting states.

Finally, in Section \ref{sec:examples}, we apply this network formulation to several examples, proving that oscillatory localization occurs on a wide variety of regular graphs, including complete graphs, complete bipartite graphs, hypercubes, square lattices of high dimension, expander graphs, and high degree graphs.

Several connections between effective resistance and classical random walks are known, such as with hitting time, commute time, and cover time \cite{Doyle1984, Tetali1991, Chandra1996}. For quantum walks, however, such connections are relatively new. Belovs \textit{et al.}\ \cite{Belovs2013} bounds the running time of a quantum walk algorithm for 3-Distinctness in terms of the resistance of a graph.


\begin{figure}
\begin{center}
	\includegraphics{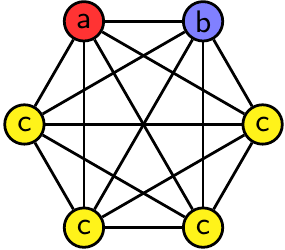}
	\caption{\label{fig:complete} The complete graph of $N = 6$ vertices. The quantum walker begins at vertex $a$, pointing at vertex $b$. By symmetry, the remaining vertices, labeled $c$, evolve identically.}
\end{center}
\end{figure}

\begin{figure}
\begin{center}
	\includegraphics{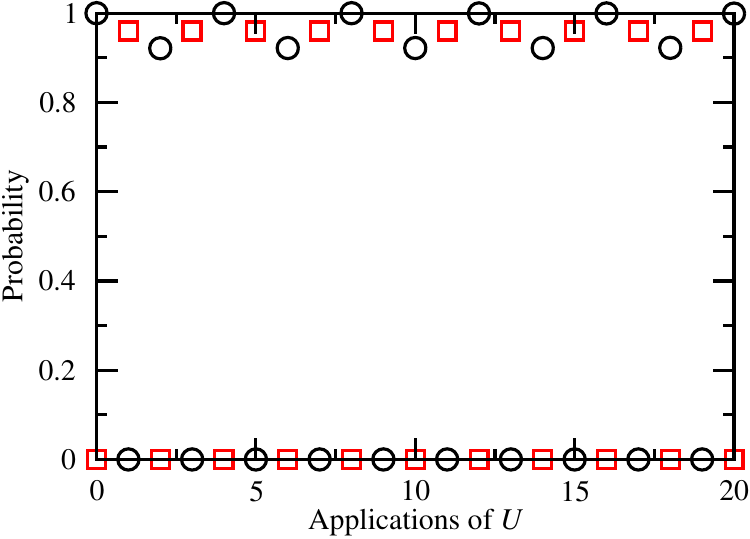}
	\caption{\label{fig:complete_prob_time_ab_ba_N16} Quantum walk on the complete graph of $N = 16$ vertices, starting in $\ket{ab}$. The probability in $\ket{ab}$ (black circles) and $\ket{ba}$ (red squares) as the quantum walk $U$ is applied.}
\end{center}
\end{figure}

\begin{table}
	\caption{\label{table:complete_prob_time_ab_ba_N16} For the quantum walk on the complete graph of $N = 16$ vertices, the probability and amplitude in $\ket{ab}$ and $\ket{ba}$.}
	\begin{ruledtabular}
	\begin{tabular}{ccccc}
		$t$ & $| \langle ab | U^t | ab  \rangle |^2$ & $| \langle ba | U^t | ab \rangle |^2$ & $\langle ab | U^t | ab \rangle$ & $\langle ba | U^t | ab \rangle$ \\
		\colrule
		$0$ & $1$ & $0$ & $1$ & $0$ \\
		$1$ & $0$ & $0.960004$ & $0$ & $-0.979798$ \\
		$2$ & $0.921608$ & $0$ & $0.960004$ & $0$ \\
		$3$ & $6.52861 \times 10^{-7}$ & $0.960004$ & $0.000807998$ & $-0.979798$ \\
		$4$ & $0.999967$ & $0$ & $0.999984$ & $0$ \\
		$5$ & $6.52329 \times 10^{-7}$ & $0.960004$ & $-0.000807669$ & $-0.979798$ \\
		$6$ & $0.921671$ & $0$ & $0.960037$ & $0$ \\
		$7$ & $2.60825 \times 10^{-6}$ & $0.960004$ & $0.00161501$ & $-0.979798$ \\
		$8$ & $0.999869$ & $0$ & $0.999935$ & $0$ \\
		$9$ & $2.60399 \times 10^{-6}$ & $0.960004$ & $-0.00161369$ & $-0.979798$ \\
		$10$ & $0.921796$ & $0$ & $0.960102$ & $0$ \\
		$11$ & $5.855 \times 10^{-6}$ & $0.960004$ & $0.00241971$ & $-0.979798$ \\
		$12$ & $0.999707$ & $0$ & $0.999853$ & $0$ \\
		$13$ & $5.84066 \times 10^{-6}$ & $0.960004$ & $-0.00241675$ & $-0.979798$ \\
		$14$ & $0.921983$ & $0$ & $0.9602$ & $0$ \\
		$15$ & $0.0000103735$ & $0.960004$ & $0.00322079$ & $-0.979798$ \\
		$16$ & $0.999479$ & $0$ & $0.999739$ & $0$ \\
		$17$ & $0.0000103396$ & $0.960004$ & $-0.00321553$ & $-0.979798$ \\
		$18$ & $0.922233$ & $0$ & $0.96033$ & $0$ \\
		$19$ & $0.0000161359$ & $0.960004$ & $0.00401695$ & $-0.979798$ \\
		$20$ & $0.999187$ & $0$ & $0.999593$ & $0$ \\
	\end{tabular}
	\end{ruledtabular}
\end{table}

\section{\label{sec:complete} Localization on the Complete Graph}

We begin with a simple example of oscillatory localization on the complete graph of $N$ vertices, an example of which is shown in Fig.~\ref{fig:complete}. As depicted in the figure, the walker is initially located at a single vertex, labeled $a$, and points towards another vertex, labeled $b$. That is, the initial state of the system is $\ket{a} \otimes \ket{a \to b}$. By the symmetry of the quantum walk, all the other vertices will evolve identically; let us call them $c$ vertices. Grouping them together, we obtain a 7D subspace for the evolution of the system:
\begin{align*}
	\ket{ab} &= \ket{a} \otimes \ket{a \to b} \\
	\ket{ac} &= \ket{a} \otimes \frac{1}{\sqrt{N-2}} \sum_{c} \ket{a \to c} \\
	\ket{ba} &= \ket{b} \otimes \ket{b \to a} \\
	\ket{bc} &= \ket{b} \otimes \frac{1}{\sqrt{N-2}} \sum_{c} \ket{b \to c} \\
	\ket{ca} &= \frac{1}{\sqrt{N-2}} \sum_{c} \ket{c} \otimes \ket{c \to a} \\
	\ket{cb} &= \frac{1}{\sqrt{N-2}} \sum_{c} \ket{c} \otimes \ket{c \to b} \\
	\ket{cc} &= \frac{1}{\sqrt{N-2}} \sum_{c} \ket{c} \otimes \frac{1}{\sqrt{N-3}} \sum_{c' \sim c} \ket{c \to c'}.
\end{align*}
So the system begins in $\ket{ab}$. The system evolves by repeated applications of the quantum walk operator \eqref{eq:U}, which in the $\{ \ket{ab}, \ket{ac}, \ket{ba}, \ket{bc}, \ket{ca}, \ket{cb}, \ket{cc} \}$ basis is
\begin{widetext}
\[ U = \begin{pmatrix}
	0 & 0 & -\frac{N-3}{N-1} & \frac{2\sqrt{N-2}}{N-1} & 0 & 0 & 0 \\
	0 & 0 & 0 & 0 & -\frac{N-3}{N-1} & \frac{2}{N-1} & \frac{2\sqrt{N-3}}{N-1} \\
	-\frac{N-3}{N-1} & \frac{2\sqrt{N-2}}{N-1} & 0 & 0 & 0 & 0 & 0 \\
	0 & 0 & 0 & 0 & \frac{2}{N-1} & -\frac{N-3}{N-1} & \frac{2\sqrt{N-3}}{N-1} \\
	\frac{2\sqrt{N-2}}{N-1} & \frac{N-3}{N-1} & 0 & 0 & 0 & 0 & 0 \\
	0 & 0 & \frac{2\sqrt{N-2}}{N-1} & \frac{N-3}{N-1} & 0 & 0 & 0 \\
	0 & 0 & 0 & 0 & \frac{2\sqrt{N-3}}{N-1} & \frac{2\sqrt{N-3}}{N-1} & \frac{N-5}{N-1} \\
\end{pmatrix}. \]
\end{widetext}
This matrix can be obtained by explicit calculation, or by using Eq.~(9) of \cite{Wong18}.

In Fig.~\ref{fig:complete_prob_time_ab_ba_N16}, we plot the probability in $\ket{ab}$ (black circles) and $\ket{ba}$ (red squares) as the quantum walk evolves, and we see that the system roughly oscillates between the two states. This is an example of oscillatory localization. As the number of vertices $N$ increases, the probability in each state at even and odd times goes to 1, so the oscillation becomes more and more certain. The precise numerical values for these probabilities are shown in Table~\ref{table:complete_prob_time_ab_ba_N16}, along with their corresponding amplitudes.

To prove that this oscillatory localization between $\ket{ab}$ and $\ket{ba}$ persists for all time, we express the initial state $\ket{ab}$ in terms of the eigenvectors and eigenvalues of $U$. The (unnormalized) eigenvectors and eigenvalues are
\begin{widetext}
\begin{alignat*}{2}
	\psi_1 &= \left( \frac{1}{N},\ 0,\ -\frac{N-3}{N(N-1)},\ \frac{2 \sqrt{N-2}}{N(N-1)},\ \frac{2 \sqrt{N-2}}{N(N-1)},\ 0,\ \frac{\sqrt{(N-2)(N-3)}}{N(N-1)} \right)^\intercal \!\!\! , \quad &1 \\
	\psi_1' &= \frac{(N-2)(N-3)}{2N(N-1)} \left( 1,\ \frac{-1}{\sqrt{N-2}},\ -1,\ \frac{1}{\sqrt{N-2}},\ \frac{1}{\sqrt{N-2}},\ \frac{-1}{\sqrt{N-2}},\ 0 \right)^\intercal \!\!\! , \quad &1 \\
	\psi_{-1} &= \frac{N-3}{2(N-1)} \left( 1,\ \frac{-1}{\sqrt{N-2}},\ 1,\ \frac{-1}{\sqrt{N-2}},\ \frac{-1}{\sqrt{N-2}},\ \frac{-1}{\sqrt{N-2}},\ \frac{2}{\sqrt{(N-2)(N-3)}} \right)^\intercal \!\!\! , &\hspace{.5cm}-1 \\
	\psi_+ &= \frac{1}{N[(N-2) - i\sqrt{N(N-2)}]} \Bigg( 1 - i\sqrt{N(N-2)},\ -i (N-3) \sqrt{N},\ N-3, \\
		&\quad\quad\quad -2 \sqrt{N-2},\ (N-2)^{3/2} + i \sqrt{N},\ 0,\ -\sqrt{N-3} (\sqrt{N-2} + i\sqrt{N}) 
	\Bigg)^\intercal \!\!\! , &e^{i\theta} \\
	\psi_+' &= \frac{N-3}{N[(N-2)-i\sqrt{N(N-2)}]} \Bigg( 1,\ \frac{1}{2} \left( \sqrt{N-2} + i\sqrt{N} \right),\ -1, \\
		&\quad\quad\quad \frac{-1}{2} \left( \sqrt{N-2} + i\sqrt{N} \right),\ \frac{-1}{2} \left( \sqrt{N-2} - i\sqrt{N} \right),\ \frac{1}{2} \left( \sqrt{N-2} - i\sqrt{N} \right),\ 0
	\Bigg)^\intercal \!\!\! , &e^{i\theta} \\
	\psi_- &= \frac{1}{N[(N-2) + i\sqrt{N(N-2)}]} \Bigg( 1 + i\sqrt{N(N-2)},\ i (N-3) \sqrt{N},\ N-3, \\
		&\quad\quad\quad -2 \sqrt{N-2},\ (N-2)^{3/2} - i \sqrt{N},\ 0,\ -\sqrt{N-3} (\sqrt{N-2} - i\sqrt{N}) 
	\Bigg)^\intercal \!\!\! , &e^{-i\theta} \\
	\psi_-' &= \frac{N-3}{N[(N-2)+i\sqrt{N(N-2)}]} \Bigg( 1,\ \frac{1}{2} \left( \sqrt{N-2} - i\sqrt{N} \right),\ -1, \\
		&\quad\quad\quad \frac{-1}{2} \left( \sqrt{N-2} - i\sqrt{N} \right),\ \frac{-1}{2} \left( \sqrt{N-2} + i\sqrt{N} \right),\ \frac{1}{2} \left( \sqrt{N-2} + i\sqrt{N} \right),\ 0
	\Bigg)^\intercal \!\!\! , & e^{-i\theta}
\end{alignat*}
\newline 
\end{widetext}
where
\[ \cos\theta = \frac{-1}{N-1}, \quad \sin\theta = \frac{-\sqrt{N(N-2)}}{N-1}. \]
Note that $\theta \in (\pi,3\pi/2)$. Alternatively, if it were defined in $(0,\pi/2)$, then the last four eigenvalues would have explicit minus signs.

It is straightforward to verify that the initial state $\ket{ab} = ( 1, 0, 0, 0, 0, 0, 0 )^\intercal$ is the sum of these unnormalized eigenvectors, i.e.,
\[ \ket{ab} = \psi_1 + \psi_1' + \psi_{-1} + \psi_+ + \psi_+' + \psi_- + \psi_-'. \]
Then the state of the system after $t$ applications of $U$ is
\begin{align*}
	U^t \ket{ab} 
	&= \psi_1 + \psi_1' + (-1)^t \psi_{-1} + e^{i\theta t} \left( \psi_+ + \psi_+' \right) \\
	& \quad + e^{-i\theta t} \left( \psi_- + \psi_-' \right).
\end{align*}
We can work out the amplitude of this in $\ket{ab}$ and $\ket{ba}$. Beginning with $\ket{ab}$:
\begin{align}
	\braket{ ab | U^t | ab } 
		&= \frac{1}{N} + \frac{(N-2)(N-3)}{2N(N-1)} + (-1)^t \frac{N-3}{2(N-1)} \notag \\
		& \quad + e^{i\theta t} \frac{1}{N} + e^{-i\theta t} \frac{1}{N} \notag \\
		&= \frac{N^2 - 3N + 4}{2N(N-1)} + (-1)^t \frac{N-3}{2(N-1)} \notag \\
		& \quad + \frac{2}{N} \cos(\theta t) \notag \\
		&= \begin{cases}
			\displaystyle \frac{N-2}{N} + \frac{2}{N} \cos(\theta t), & t \text{ even} \\
			\displaystyle \frac{2}{N(N-1)} + \frac{2}{N} \cos(\theta t), & t \text{ odd} \\
		\end{cases} \label{eq:complete_ab}.
\end{align}
Now for $\ket{ba}$:
\begin{align}
	\braket{ ba | U^t | ab } 
		&= -\frac{N-3}{N(N-1)} - \frac{(N-2)(N-3)}{2N(N-1)} \notag \\
		& \quad + (-1)^t \frac{N-3}{2(N-1)} + 0 + 0 \notag \\
		&= \begin{cases}
			\displaystyle 0, & t \text{ even} \\
			\displaystyle -\frac{N-3}{N-1}, & t \text{ odd} \\
		\end{cases} \label{eq:complete_ba}.
\end{align}
Using these formulas with $N = 16$, we get exactly the amplitudes in $\ket{ab}$ and $\ket{ba}$ in Table~\ref{table:complete_prob_time_ab_ba_N16}. Furthermore, for large $N$, we get that the amplitude in $\ket{ab}$ roughly alternates between $1$ and $0$, while the amplitude in $-| ba \rangle$ roughly alternates between $0$ and $1$. So for large $N$, the system alternates between being in $\ket{ab}$ and $\ket{ba}$ with probability nearly $1$. This is an example of oscillatory localization, and it persists for all time.


\section{\label{sec:exact} Exact Oscillatory States: 1-Eigenvectors of $U^2$}

While the above example of oscillatory localization on the complete graph was simple enough to be exactly analyzed, for general graphs we expect such analysis to be intractable. Here we begin to develop a more general theory for determining when oscillatory localization occurs.

To start, we observe that oscillatory localization implies that the state of the system returns to itself after two applications of the quantum walk $U$ \eqref{eq:U}. In other words, states that \emph{exactly} oscillate are eigenvectors of $U^2$ with eigenvalue $1$. In this section, we find these eigenstates, assuming that the graph $G$ is $d$-regular, connected, and undirected. We show that there are only two distinct types of $1$-eigenvectors of $U^2$, which we call uniform states and flip states, and that any $1$-eigenvector of $U^2$ can be expressed in terms of them.

In general, although the $1$-eigenvectors of $U^2$ oscillate between two states, they may not be localized. We disregard this for the time being, finding oscillatory states regardless of their spatial distributions. Later in the paper, we project localized initial states, such as $\ket{ab}$ from the complete graph in the last section, onto these general oscillatory states and prove oscillatory localization on a large variety of graphs.


\subsection{Uniform States}

The first type of eigenvector of $U^2$ with eigenvalue $1$ is a \emph{uniform state}. To define it, we start by introducing some notation. 

\begin{definition}
	Let $V$ be the vertex set of $G$. For a vertex subset $T \subseteq V$, define the outer product
	\[ \ket{\sigma_T} = \ket{s_T} \otimes \ket{s_c}, \]
	where
	\[ \ket{s_T} = \frac{1}{\sqrt{|T|}}\sum_{t \in T} \ket{t} \quad \text{and} \quad \ket{s_c} = \frac{1}{\sqrt{d}} \sum_{i=1}^d \ket{i} \]
	are equal superpositions over the vertices and directions, respectively. Evaluating the tensor product,
	\[ \ket{\sigma_T} = \frac{1}{\sqrt{d|T|}}\sum_{\substack{t \in T \\ v \sim t}} \ket{tv}, \]
	where $\ket{tv}$ denotes the quantum walker at vertex $t$, pointing towards vertex $v$.
\end{definition}

\noindent Using this notation $\ket{\sigma_T}$, we define the uniform states of $G$, depending on whether $G$ is non-bipartite or bipartite:

\begin{definition}
	The \emph{uniform states of $G$} are
	\begin{itemize}
		\item	$\ket{\ustate_V}$ if $G$ is non-bipartite (i.e., the uniform superposition over all the vertices and directions).
		\item	$\ket{\ustate_X}$, $\ket{\ustate_Y}$, and their linear combinations, if $G$ is bipartite, where $X$ and $Y$ are the partite sets of $G$ (i.e., the uniform superpositions over each partite set and all directions, and linear combinations of them).
	\end{itemize}
\end{definition}

\noindent Now we give a simple proof showing that uniform states, defined above, are indeed $1$-eigenvectors of $U^2$.

\begin{lemma}
	The uniform states are eigenvectors of $U^2$ with eigenvalue 1.
\end{lemma}

\begin{proof}
	For an arbitrary $T \subseteq V$, we have
	\begin{align*}
		U \ket{\ustate_T} 
			&= S (I_N \otimes C) (\ket{s_T} \otimes \ket{s_c}) = S (\ket{s_T} \otimes C\ket{s_c}) \\
			&= S (\ket{s_T} \otimes \ket{s_c}).
	\end{align*}
	Let us consider separately the cases when $G$ is non-bipartite or bipartite.
	
	When $G$ is non-bipartite, consider $\ket{\ustate_V}$. We have
	\begin{align*}
		U \ket{\ustate_V}
		&= S (\ket{s_V} \otimes \ket{s_c}) \\
		&= S\left(\frac{1}{2\sqrt{dN}}\sum_{\substack{v \in V \\ u \sim v}}(\ket{uv} + \ket{vu})\right) \\
		&= \frac{1}{2\sqrt{dN}}\sum_{\substack{v \in V \\ u \sim v}} (\ket{vu} + \ket{uv}) \\
		&= \ket{\ustate_V}.
	\end{align*}
	Therefore, $\ket{s_V}$ is an eigenvector of $U$ with eigenvalue 1, and hence it is also an eigenvector of $U^2$ with eigenvalue 1.

	Now suppose $G$ is bipartite. For $\ket{\ustate_X}$, we have 
	\begin{align*}
		U\ket{\ustate_X}
		&= S (\ket{s_X} \otimes \ket{s_c}) \\
		&= S\left(\frac{1}{\sqrt{dN/2}}\sum_{\substack{x \in X \\ y \sim x}}\ket{xy}\right) \\
		&= \frac{1}{\sqrt{dN/2}}\sum_{\substack{y \in Y \\ x \sim y}}\ket{yx} \\
		&= \ket{\ustate_Y}.
	\end{align*}
	Thus in two steps, $\ket{\ustate_X}$ evolves to $\ket{\ustate_Y}$ and then back to $\ket{\ustate_X}$. Therefore, $\ket{\ustate_X}$ is an eigenvector of $U^2$ with eigenvalue 1. The same holds for $\ket{\ustate_Y}$.
\end{proof}


\subsection{Flip States}

The second type of eigenvector of $U^2$ with eigenvalue $1$ is called a \emph{flip state}. To define it, we first introduce the average amplitude of the edges pointing out from a vertex and into a vertex. Note that each undirected edge $\{u, v\}$ actually consists of two amplitudes: one for the walker at $u$ pointing towards $v$, and one for the walker at $v$ pointing towards $u$. So we can formulate each undirected edge as two directed edges $(u, v)$ and $(v, u)$ that are described by the states $\ket{uv}$ and $\ket{vu}$ \cite{Hillery2003}. Then we have the following definition:

\begin{definition}
	For a state $\ket{\psi}$, define the \emph{average outgoing and incoming amplitudes} at a vertex $u$ as
	\[ \avgout{u}{\psi} = \frac{1}{d} \sum_{v \sim u} \braket{uv | \psi} \quad \text{and} \quad \avgin{u}{\psi} = \frac{1}{d} \sum_{v \sim u} \braket{vu | \psi}. \]
\end{definition}

\noindent These averages are important because the Grover coin $C$ performs the ``inversion about the average'' of Grover's algorithm \cite{Grover1996}, as explained in the following lemma:

\begin{lemma}
	\label{lemma:Grover}
	Consider a general coin state
	\[ \ket{\psi_c} = \sum_{i=1}^d \alpha_i \ket{i}. \]
	The Grover coin $C = 2 \ketbra{s_c}{s_c} - I_d$ inverts each amplitude $\alpha_i$ about the average amplitude $\overline{\psi_c} = \frac{1}{d} \sum_{j=1}^d \alpha_j$, i.e.,
	\[ C \ket{\psi_c} = \sum_{i=1}^d ( 2 \overline{\psi_c} - \alpha_i ) \ket{i}. \]
\end{lemma}

\noindent Additionally note that if the average amplitude is zero (i.e., $\overline{\psi_c} = 0$), then $C \ket{\psi_c} = -\ket{\psi_c}$.

\begin{proof}
	\begin{align*}
		C \ket{\psi_c}
			&= \left( 2 \ketbra{s_c}{s_c} - I_d \right) \ket{\psi_c} \\
			&= 2 \ket{s_c} \braket{s_c | \psi_c} - \ket{\psi_c} \\
			&= 2 \cdot \frac{1}{\sqrt{d}} \sum_{i=1}^d \ket{i} \frac{1}{\sqrt{d}} \sum_{j=1}^d \alpha_j - \ket{\psi_c} \\
			&= 2 \left( \frac{1}{d} \sum_{j=1}^d \alpha_j \right) \sum_{i=1}^d \ket{i}  - \ket{\psi_c} \\
			&= 2 \overline{\psi_c} \sum_{i=1}^d \ket{i}  - \ket{\psi_c} \\
			&= \sum_{i=1}^d ( 2 \overline{\psi_c} - \alpha_i ) \ket{i}.
	\end{align*}
\end{proof}

\noindent Now a flip state can be defined in terms of these average outgoing and incoming amplitudes:

\begin{definition}
	\label{def:flip-state}
	We say that a state $\ket{\fstate}$ is a \emph{flip state} if for each vertex $u$ we have
	\begin{equation}
		\label{eq:zero-sum}
		\avgout{u}{\psi} = 0 \quad \text{and} \quad \avgin{u}{\psi} = 0.
	\end{equation}
	In other words, the sum of the amplitudes of the edges starting at $u$ is $0$, and the sum of the amplitudes of the edges ending at $u$ is $0$, for every vertex in the graph.
\end{definition}

\noindent To explain why this is called a flip state, we introduce the concept of negating and flipping the amplitudes between each pair of vertices:

\begin{definition}
	For a state $\ket{\psi}$, define the \emph{flipped state} $\ket{\flip{\psi}}$ to be such that
	\[ \braket{uv | \flip{\psi}} = -\braket{vu | \psi} \]
	for any edge $(u, v)$.
\end{definition}

\noindent A short lemma and proof shows that flip states are equal to their flipped versions under action by $U$, hence their name:

\begin{lemma}
	\label{thm:flip-step}
	Let $\ket{\fstate}$ be a flip state. Then $U\ket{\fstate}=\ket{\flip{\fstate}}$.
\end{lemma}

\begin{proof}
	We have
	\begin{align*}
		U\ket{\fstate}
		&= S (I_N \otimes C) \ket{\fstate} \\
		&= S (I_N \otimes C) \sum_{(u,v)} \braket{uv | \fstate} \ket{uv} \\
		&= S \left[ \sum_{(u,v)} \left( 2\avgout{u}{\fstate} - \braket{uv | \fstate} \right) \ket{uv} \right] \\
		&= S \left[  \sum_{(u,v)} -\braket{uv | \fstate}  \ket{uv} \right] \\
		&= \sum_{(u,v)} - \braket{uv | \fstate} \ket{vu} \\
		&= \ket{\flip{\fstate}}.
	\end{align*}
	Note that the sum is over all directed edges $(u,v)$, and the third line is obtained from Lemma \ref{lemma:Grover} since the Grover coin inverts about the average. The fourth line is due to $\ket{\fstate}$ being a flip state, so $\avgout{u}{\fstate} = 0$.
\end{proof}

\noindent Note that the flipped version of a flip state is also a flip state:

\begin{lemma}
	\label{thm:flip-flip}
	Let $\ket{\fstate}$ be a flip state. Then $\ket{\flip{\fstate}}$ is also a flip state.
\end{lemma}

\begin{proof}
	For any vertex $u$, $\avgout{u}{\flip{\fstate}}=-\avgin{u}{\fstate}=0$ and $\avgin{u}{\flip{\fstate}}=-\avgout{u}{\fstate} =0$.
\end{proof}

\noindent Now it is straightforward to prove that flip states are $1$-eigenvectors of $U^2$, so they oscillate between two states.

\begin{lemma}
	Any flip state $\ket{\fstate}$ is an eigenvector of $U^2$ with eigenvalue 1.
\end{lemma}

\begin{proof}
	Since the flipped state $\ket{\flip{\fstate}}$ is also a flip state by Lemma \ref{thm:flip-flip},
	\[ U^2\ket{\fstate} = U\ket{\flip{\fstate}} = \ket{\fstate} \]
	by Lemma \ref{thm:flip-step}. Therefore $\ket{\fstate}$ is an eigenvector of $U^2$ with eigenvalue 1.
\end{proof}


\subsection{Expansion of Eigenvectors}

In the last two sections, we defined uniform states and flip states, both of which are eigenvectors of $U^2$ with eigenvalue $1$. In this section, we work towards a theorem that proves that all $1$-eigenvectors of $U^2$ can be written as linear combinations of uniform states and/or flip states. That is, uniform states and flip states form a complete basis for states exhibiting exact oscillations between two states. Towards this goal, we first prove some general properties of the $1$-eigenvectors of $U^2$:

\begin{lemma}
	\label{thm:avg-step}
	Let $\ket{\psi}$ be an eigenvector of $U^2$ with eigenvalue 1, and denote $\ket{\psi'} = U\ket{\psi}$. Suppose $u$ and $v$ are connected by an edge in $G$. Then
	\[ \avgout{u}{\psi} = \avgout{v}{\psi'}. \]
\end{lemma}

\begin{proof}
	Denote $\braket{uv | \psi} = \delta$. Let us examine the effect of two steps of the quantum walk \eqref{eq:U} on this amplitude, as shown in Fig.~\ref{fig:evolution}.

	\begin{figure*}
	\begin{center}
		\subfloat[]{
			\includegraphics{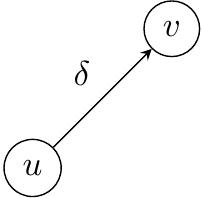}
		} \quad \quad \quad
		\subfloat[]{
			\includegraphics{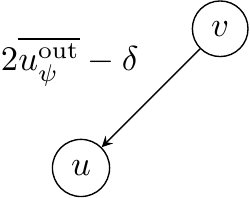}
		} \quad \quad \quad
		\subfloat[]{
			\includegraphics{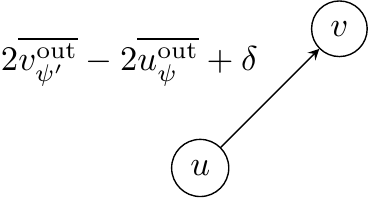}
		}
		\caption{\label{fig:evolution} Evolution of a single amplitude shown in (a), after one application of $U$ in (b), and two applications of $U$ in (c).}
	\end{center}
	\end{figure*}

	After the first step, we have
	\[ \braket{ vu | \psi'} = \braket{ vu | U | \psi } = 2\avgout{u}{\psi} - \delta, \]
	since the Grover coin causes the amplitude to be inverted about the average (Lemma \ref{lemma:Grover}), and then the flip-flop shift causes the edge to switch directions. After the second step, we get
	\begin{align*}
		\braket{ uv | U^2 | \psi } 
			&= \braket{ uv | U | \psi' } = 2\avgout{v}{\psi'} - \braket{ vu | \psi' } \\
			&= 2\avgout{v}{\psi'} - 2\avgout{u}{\psi} + \delta.
	\end{align*}
	Since $\ket{\psi}$ is an eigenvector of $U^2$ with eigenvalue 1, we have $\braket{uv | U^2 | \psi} = \braket{uv | \psi}$.
	Hence $2\avgout{v}{\psi'} - 2\avgout{u}{\psi} = 0$ and $\avgout{v}{\psi'} = \avgout{u}{\psi}$.
\end{proof}

\begin{lemma}
	\label{thm:avg}
	Let $\ket{\psi}$ be an eigenvector of $U^2$ with eigenvalue 1. Suppose $u$ and $v$ are two vertices (not necessarily distinct) of $G$.

	(a) If there exists a walk of even length between $u$ and $v$ in $G$, then
	\[ \avgout{u}{\psi} = \avgout{v}{\psi}. \]

	(b) If there exists a walk of odd length between $u$ and $v$ in $G$, then
	\[ \avgout{u}{\psi} = \avgin{v}{\psi}. \]
\end{lemma}

\begin{proof}
	Denote $\ket{\psi'} = U\ket{\psi}$.

	(a) Suppose $\{u, t\}$ and $\{t, w\}$ are edges of $G$. By Lemma~\ref{thm:avg-step},
	\[ \avgout{u}{\psi} = \avgout{t}{\psi'} = \avgout{w}{\psi}. \]
	Then the statement holds by transitivity.

	(b) Suppose there is an edge $\{u, t\}$. By Lemma~\ref{thm:avg-step},
	\begin{align*}
		\avgout{u}{\psi}
		&= \avgout{t}{\psi'} \\
		&= \frac{1}{d} \sum_{w \sim t} \braket{tw | \psi'} \\
		&= \frac{1}{d} \sum_{w \sim t} \langle tw | S(I_N \otimes C) | \psi \rangle \\
		&= \frac{1}{d} \sum_{w \sim t} \langle wt | (I_N \otimes C) | \psi \rangle \\
		&= \frac{1}{d} \sum_{w \sim t} \left( 2\avgout{w}{\psi} - \braket{wt | \psi} \right) \\
		&= \frac{1}{d} \sum_{w \sim t} \left( 2\avgout{u}{\psi} - \braket{wt | \psi} \right) \quad \text{by (a)} \\
		&= 2\avgout{u}{\psi} - \frac{1}{d} \sum_{w \sim t} \braket{wt | \psi} \\
		&= 2\avgout{u}{\psi} - \avgin{t}{\psi}.
	\end{align*}
	Hence $\avgout{u}{\psi} = \avgin{t}{\psi}$. The statement holds by transitivity once again.
\end{proof}

With these lemmas in place, we are now able to prove the main result of this section, that exact oscillatory states are composed entirely of uniform states and/or flip states.

\begin{theorem}
	\label{thm:expansion}
	Let $\ket{\psi}$ be an eigenvector of $U^2$ with eigenvalue 1.
	Then for some flip state $\ket{\fstate}$,

	(a) if $G$ is non-bipartite,
	\[ \ket{\psi} = \fcoef \ket{\fstate} + \ucoef_V \ket{\ustate_V}, \]

	(b) if $G$ is bipartite,
	\[ \ket{\psi} = \fcoef \ket{\fstate} + \ucoef_X \ket{\ustate_X} + \ucoef_Y \ket{\ustate_Y}. \]
\end{theorem}

\begin{proof}
	We prove each case separately:

	(a) Let $u$ and $v$ be two arbitrary vertices of $G$. Since $G$ is not bipartite, it contains a cycle of odd length. Therefore there exists both a walk of even length and a walk of odd length between $u$ and $v$, as $G$ is a connected graph. Thus $\avgout{u}{\psi} = \avgout{v}{\psi}$ and $\avgout{u}{\psi} = \avgin{v}{\psi}$ by Lemma \ref{thm:avg}. Therefore $\avgout{v}{\psi} = \avgin{v}{\psi} = \overline{\psi}$ for any vertex $v$.

	Consider the (likely unnormalized) state $\ket{\varphi} = \ket{\psi} - \ket{\ustate_V} \braket{\ustate_V | \psi}$. We prove that, up to normalization, it is a flip state by checking the two conditions of \eqref{eq:zero-sum}. First, note that
	\begin{align}
		\braket{\ustate_V | \psi} \label{eq:B1}
		&= \sum_{u \in V} \sum_{v \sim u} \frac{1}{\sqrt{dN}} \braket{uv | \psi} \\
		&= \sqrt{\frac d N} \sum_{u \in V} \sum_{v \sim u} \frac{1}{d} \braket{uv | \psi} \notag \\
		&= \sqrt{\frac d N} \sum_{u \in V} \avgout{u}{\psi} \notag \\
		&= \sqrt{\frac d N} N \overline{\psi} \notag \\
		&= \sqrt{dN}\,\overline{\psi}. \label{eq:B2}
	\end{align}
	On the other hand, $\braket{uv | \ustate_V} = \frac{1}{\sqrt{dN}}$ for any pair of adjacent vertices $u$ and $v$. Hence,
	\begin{align}
		\avgout{u}{\varphi} \label{eq:A1}
		&= \frac{1}{d} \sum_{v \sim u} \braket{uv | \varphi} \\
		&= \frac{1}{d} \sum_{v \sim u} \left( \braket{uv | \psi} - \braket{uv | \ustate_V} \braket{\ustate_V | \psi} \right) \notag \\
		&= \frac{1}{d} \sum_{v \sim u} \left( \braket{uv | \psi} - \frac{1}{\sqrt{dN}}\sqrt{dN} \, \overline{\psi} \right) \notag \\
		&= \frac{1}{d} \sum_{v \sim u} \left( \braket{uv | \psi} - \overline{\psi} \right) \notag \\
		&= \avgout{u}{\psi} - \overline{\psi} \notag \\
		&= 0. \label{eq:A2}
	\end{align}
	This is the first condition of \eqref{eq:zero-sum}. The second condition is proved similarly to equations \eqref{eq:A1}--\eqref{eq:A2}:
	\begin{align*}
		\avgin{u}{\varphi}
		&= \frac{1}{d} \sum_{v \sim u} (\braket{vu | \psi} - \braket{vu | \ustate_V} \braket{\ustate_V | \psi}) \\
		&= \avgin{u}{\psi} - \overline{\psi} \\
		&= 0.
	\end{align*}
	So $\ket{\varphi}$ is a flip state.

	(b) Let the partite sets be $X$ and $Y$. Let $\overline{\psi_T} = \frac{1}{|T|} \sum_{t \in T} \avgout{t}{\psi}$. Suppose $x \in X$ and $y \in Y$. Then we have the following properties by Lemma \ref{thm:avg}:
	\[ \avgout{x}{\psi} = \avgin{y}{\psi} = \overline{\psi_X} \quad \text{and} \quad \avgout{y}{\psi} = \avgin{x}{\psi} = \overline{\psi_Y}. \]
	Then similarly to (a), we can prove that the state $\ket{\varphi} = \ket{\psi} - \ket{\ustate_X} \braket{\ustate_X | \psi} - \ket{\ustate_Y} \braket{\ustate_Y | \psi}$ is a flip state by checking the two conditions of \eqref{eq:zero-sum}. First, note that similarly to equations \eqref{eq:B1}--\eqref{eq:B2},
	\begin{align*}
		\braket{\ustate_X | \psi}
		&= \sum_{x \in X} \sum_{y \sim X} \frac{1}{\sqrt{dN/2}} \braket{xy | \psi} \\
		&= \sqrt{dN/2}\,\overline{\psi_X}.
	\end{align*}
	On the other hand, $\braket{xy | \ustate_X} = \frac{1}{\sqrt{dN/2}}$ for any edge $(x, y)$. Moreover, $\braket{xy | \ustate_Y} = 0$. Hence,
	\begin{align}
		\avgout{x}{\varphi} \label{eq:C1}
		&= \frac{1}{d} \sum_{y \sim x} \Big( \braket{xy | \psi} - \braket{xy | \ustate_X} \braket{\ustate_X | \psi} \\
		&\quad\quad\quad\quad\quad - \braket{xy | \ustate_Y} \braket{\ustate_Y | \psi} \Big) \notag \\
		&= \frac{1}{d} \sum_{y \sim x} \Bigg( \braket{xy | \psi} - \frac{1}{\sqrt{dN/2}}\sqrt{dN/2} \, \overline{\psi_X} \notag \\
		&\quad\quad\quad\quad\quad - 0\cdot \braket{\ustate_Y | \psi} \Bigg) \notag \\
		&= \frac{1}{d} \sum_{y \sim x} \left(\braket{xy | \psi} - \overline{\psi_X}\right) \notag \\
		&= \avgout{x}{\psi} - \overline{\psi_X} \notag \\
		&= 0. \label{eq:C2}
	\end{align}
	This is the first condition of \eqref{eq:zero-sum}. The second condition comes similarly to equations \eqref{eq:C1}--\eqref{eq:C2},
	\begin{align*}
		\avgin{y}{\varphi}
		&= \frac{1}{d} \sum_{x \sim y} \Big( \braket{xy | \psi} - \braket{xy | \ustate_X} \braket{\ustate_X | \psi} \\
		&\quad\quad\quad\quad\quad - \braket{xy | \ustate_Y} \braket{\ustate_Y | \psi} \Big) \\
		&= \avgin{y}{\psi} - \overline{\psi_X} \\
		&= 0.
	\end{align*}

	Similarly, we prove that $\avgout{y}{\varphi} = 0$ and $\avgin{x}{\varphi} = 0$. So $\ket{\varphi}$ is a flip state.
\end{proof}

We end this section by showing that uniform states and flip states are orthogonal to each other, so they serve as an orthonormal basis  for 1-eigenvectors of $U^2$.

\begin{lemma}
	\label{thm:flip-uniform}
	Any flip state $\ket{\fstate}$ is orthogonal to any state $\ket{\ustate_T}$:
\end{lemma}

\begin{proof}
	\begin{align*}
		\braket{\fstate | \ustate_T}
		&= \sum_{t \in T} \sum_{v \sim t} \braket{\fstate | tv} \braket{tv | \ustate_T} \\
		&= \sum_{t \in T} \frac{1}{\sqrt{d|T|}} \sum_{v \sim t} \braket{\fstate | tv} \\
		&= \sum_{t \in T} \sqrt\frac{d}{|T|} \cdot \avgout{t}{\fstate} \\
		&= 0.
	\end{align*}
\end{proof}

Also, the states $\ket{\ustate_X}$ and $\ket{\ustate_Y}$ are orthogonal because for any edge $\ket{uv}$, only one of them has a non-zero amplitude at this edge. Thus flip states and uniform states form a complete orthogonal basis for the eigenvectors of $U^2$ with eigenvalue 1.


\section{\label{sec:approximate} Approximate Oscillatory States}

In the previous section, we found the states that \emph{exactly} exhibit oscillation between two quantum states, returning to themselves after two applications of $U$. We showed that these states are spanned by uniform states and flip states. In our example of oscillatory localization on the complete graph, however, we showed that the system \emph{approximately} alternated between $\ket{ab}$ and $-| ba \rangle$. So, while these states are not exact $1$-eigenvectors of $U^2$, they are ``close enough'' that oscillatory localization still occurs.

In this section, we give conditions for when a starting state $\ket{\psi_0}$ is ``close enough'' to being a $1$-eigenvector of $U^2$ that it exhibits oscillations. We do this by expanding $\ket{\psi_0}$ as a linear combination of flip states, uniform states, and whatever state remains. If the overlap of $\ket{\psi_0}$ with the flip states and/or uniform states is sufficiently large, then oscillations occur. That is, the state approximately alternates between $\ket{\psi_0}$ at even steps and $U\ket{\psi_0}$ at odd steps.


\subsection{Estimate on the Oscillations}

The following theorem gives a bound on the extent of oscillations for an arbitrary starting state.

\begin{theorem}
	\label{thm:localization-bound}
	Let $\ket{\psi_0}$ be the starting state of the quantum walk. It can be expressed as
	\[ \ket{\psi_0} = \fcoef \ket{\fstate} + \ucoef \ket{\ustate} + \rcoef \ket{\rstate}, \]
	where
	\begin{itemize}
		\item	$\ket{\fstate}$ is a normalized flip state;
		\item	$\ket{\ustate}$ is a normalized uniform state, equal to $\ket{\ustate_V}$ if $G$ is non-bipartite, or a normalized linear combination of $\ket{\ustate_X}$ and $\ket{\ustate_Y}$ if $G$ is bipartite;
		\item	$\ket{\rstate}$ is some normalized ``remainder'' state orthogonal to $\ket{\fstate}$ and $\ket{\ustate}$.
	\end{itemize}
	Then
	\begin{enumerate}[label=(\alph*)]
		\item	after an even number of steps $2t$,
			\[ \left| \braket{ \psi_0 | U^{2t} | \psi_0 }\right| \geq 2\left(|\fcoef|^2 + |\ucoef|^2\right)-1, \]
		\item	after an odd number of steps $2t+1$,
			\[ \left| \braket{ \flip{\psi_0} | U^{2t+1} | \psi_0 } \right| \geq 2\max\left(|\fcoef|^2,|\ucoef|^2\right)-1. \]
	\end{enumerate}
\end{theorem}

\begin{proof}
	We prove each part separately.

	(a) After $2t$ steps, the state of the quantum walk is
	\[ U^{2t}\ket{\psi_0} = \fcoef \ket{\fstate} + \ucoef \ket{\ustate} + \rcoef U^{2t}\ket{\rstate}. \]
	Since unitary operators preserve the inner product between vectors, we have $\braket{ \fstate | U^{2t} | \rstate } = \braket{ \ustate | U^{2t} | \rstate } =0$.
	Therefore,
	\[ \braket{ \psi_0 | U^{2t} | \psi_0 } = |\fcoef|^2 + |\ucoef|^2 + |\rcoef|^2 \braket{ \rstate | U^{2t} | \rstate }. \]
	Then we have
	\begin{align*}
		\left|\braket{ \psi_0 | U^{2t} | \psi_0 }\right|
		&\geq |\fcoef|^2 + |\ucoef|^2 - |\rcoef|^2 \\
		&= |\fcoef|^2 + |\ucoef|^2- (1-|\fcoef|^2 - |\ucoef|^2)\\
		&= 2\left(|\fcoef|^2 + |\ucoef|^2\right)-1.
	\end{align*}

	(b) After $2t+1$ steps the state of the quantum walk is
	\begin{align*}
		U^{2t+1}\ket{\psi_0}
		&= \fcoef U^{2t+1}\ket{\fstate} + \ucoef U^{2t+1}\ket{\ustate} + \rcoef U^{2t+1}\ket{\rstate} \\
		&= \fcoef \ket{\flip{\fstate}} + \ucoef U\ket{\ustate} + \rcoef U^{2t+1}\ket{\rstate}.
	\end{align*}
	For a non-bipartite graph,
	\[ U\ket{\ustate} = U\ket{\ustate_V} = \ket{\ustate_V} = -\ket{\flip{\ustate_V}} = -\ket{\flip{\ustate}}. \]
	For a bipartite graph,
	\begin{align*}
		U\ket{\ustate}
			&= U(p\ket{\ustate_X}+q\ket{\ustate_Y}) = p\ket{\ustate_Y} + q\ket{\ustate_X} \\
			&= -p\ket{\flip{\ustate_X}}-q\ket{\flip{\ustate_Y}} = -\ket{\flip{\ustate}}.
	\end{align*}
	Hence, in either case,
	\[ U^{2t+1}\ket{\psi_0} = \fcoef \ket{\flip{\fstate}} - \ucoef \ket{\flip{\ustate}} + \rcoef U^{2t+1}\ket{\rstate}. \]
	The flipped starting state is
	\[ \ket{\flip{\psi_0}} = \fcoef \ket{\flip{\fstate}} + \ucoef \ket{\flip{\ustate}} + \ket{\flip{\rstate}}. \]
	To obtain the value of $\braket{ \flip{\psi_0} | U^{2t+1} | \psi_0 }$, we look at the inner products of the vectors that contribute to $\ket{\flip{\psi_0}}$ and $\ket{\psi_0}$. Since unitary operations preserve the inner product between vectors, we have that $\braket{ \flip{\fstate} | U^{2t+1} | \rstate }=\braket{\fstate | \rstate} = 0$ and $\braket{ \flip{\ustate} | U^{2t+1} | \rstate }=\braket{\ustate | \rstate} = 0$. Note that the ``flip'' transformation that takes $\ket{\psi}$ to $\ket{\flip{\psi}}$ is unitary. Hence $\braket{\flip{\rstate} | \flip{\fstate}} = \braket{\rstate | \fstate} = 0$ and $\braket{\flip{\rstate} | \flip{\ustate}} = \braket{\rstate | \ustate} = 0$. Therefore,
	\[ \braket{ \flip{\psi_0} | U^{2t+1} | \psi_0 } = |\fcoef|^2 - |\ucoef|^2 + |\rcoef|^2 \braket{ \flip{\rstate} | U^{2t} | \rstate }. \]
	Thus we have
	\begin{align*}
		\left| \braket{ \flip{\psi_0} | U^{2t+1} | \psi_0 }\right| 
		&\geq \left| |\fcoef|^2 - |\ucoef|^2  \right|- |\rcoef|^2 \\
		&=\left| |\fcoef|^2 - |\ucoef|^2 \right|- (1-|\fcoef|^2 - |\ucoef|^2) \\
		&= 2\max\left(|\fcoef|^2,|\ucoef|^2\right)-1.
	\end{align*}
\end{proof}


\subsection{Conditions for Oscillations}

With this theorem, we can determine the conditions on $|\fcoef|^2$ and $|\ucoef|^2$ for oscillations between two states to occur. If $|\fcoef|^2 + |\ucoef|^2 > \frac{1}{2}$, then after any number of even steps the quantum walk is in $\ket{\psi_0}$ with probability $\Theta(1)$. This occurs because the starting state is closer to some eigenvector of $U^2$ with eigenvalue 1 than any other eigenvector with a different eigenvalue. At odd steps, the quantum walk is in $\ket{\flip{\psi_0}}$ with probability $\Theta(1)$ if either $|\fcoef|^2 > \frac{1}{2}$ or $|\ucoef|^2 > \frac{1}{2}$. Thus in this case, the starting state should be very close to either a flip or a uniform state.

The value of $\ucoef$ is easy to explicitly calculate. If the graph is non-bipartite, then the only uniform state is $\ket{\ustate_V}$, so $\ucoef = \braket{\ustate_V | \psi_0}$. If the graph is bipartite, then there are two uniform basis states $\ket{\ustate_X}$ and $\ket{\ustate_Y}$, so $|\ucoef|^2 = |\ucoef_X|^2 + |\ucoef_Y|^2$, where $\ucoef_X = \braket{\ustate_X | \psi_0}$ and $\ucoef_Y = \braket{\ustate_Y | \psi_0}$.

In general, the value of $|\fcoef|^2$ is more difficult to find since the size of the basis for the flip states can be large. In the next section, however, we prove that $|\fcoef|^2$ can be lower bounded using power dissipation, effective resistance, and connectivity of a related classical electric circuit. Since the quantum walk alternates between $\ket{\psi_0}$ and $\ket{\flip{\psi_0}}$ with probability $\Theta(1)$ if $|\alpha|^2 > \frac{1}{2}$, these quantities of classical circuits can be used to inform whether the quantum walk oscillates between two states.


\section{\label{sec:circuits} Oscillations Using Electric Circuits}

\subsection{Network Flows and Flip States}

In this section, we show that there is a close connection between the flip states of a quantum walk on a graph and the flows in a related network, of which electric current is a special case. Recall that a quantum walk has two amplitudes on each edge $\{u,v\}$, one from vertex $u$ going to vertex $v$ and another from $v$ going to $u$. So to associate this to a network flow or current, we need to split up these two amplitudes. This can be done using the \emph{bipartite double graph} from graph theory \cite{Brouwer1989}. Given a graph $G$, its bipartite double graph $G_b$ is constructed as follows. For each vertex $v$ in $G$, there are two vertices $v_{\tout}$ and $v_{\tin}$ in $G_b$. For each edge $\{u, v\}$ in $G$, there are two edges $\{u_{\tout}, v_{\tin}\}$ and $\{v_{\tout}, u_{\tin}\}$ in $G_b$, connected as shown in Fig.~\ref{fig:edge_bipartite}. As an example, consider the complete graph of three vertices in Fig.~\ref{fig:complete3}. Applying this doubling procedure, we get its bipartite double graph in Fig.~\ref{fig:complete3_double}. Note that this bipartite double graph is a cycle, which we make evident by rearranging the graph in Fig.~\ref{fig:complete3_double_rearrange}.

\begin{figure}
\begin{center}
	\includegraphics{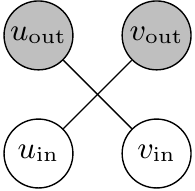}
	\caption{\label{fig:edge_bipartite} The counterpart of an edge $\{u, v\}$ of $G$ in the bipartite double graph $G_b$.}
\end{center}
\end{figure}

\begin{figure}
\begin{center}
	\subfloat[]{
		\label{fig:complete3}
		\includegraphics{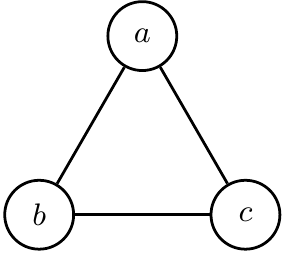}
	} \quad \quad
	\subfloat[]{
		\label{fig:complete3_double}
		\includegraphics{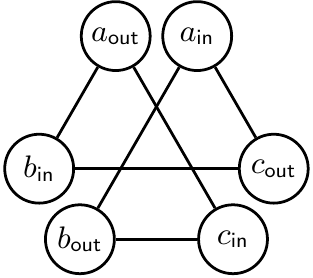}
	}

	\subfloat[]{
		\label{fig:complete3_double_rearrange}
		\includegraphics{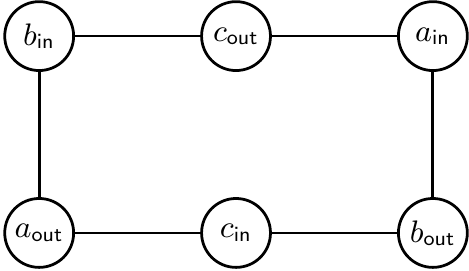}
	}
	\caption{(a) $K_3$, the complete graph of three vertices. (b) The bipartite double graph of $K_3$. (c) Rearrangement of the bipartite double graph of $K_3$.}
\end{center}
\end{figure}

Now let us define a network flow in a graph, which is another concept from graph theory \cite{Bollobas1979}. Let $G$ be a graph and $E(G)$ its edge set. Denote by $\vec{E}(G) = \{(u, v) \mid \{u, v\} \in E(G)\}$ the set of directed edges of $G$. A \emph{network flow} is a function $f : \vec{E}(G) \rightarrow \mathbb C$ that assigns a certain amount of flow passing through each edge \cite{Bollobas1979}. A network flow is called a \emph{circulation} if it satisfies two properties: (a) Skew symmetry:~the flow on an edge $(u, v)$ is equal to the negative flow on the reversed edge $(v, u)$, \textit{i.e.}, $f(u, v) = -f(v, u)$, and (b) Flow conservation:~the amount of the incoming flow at a vertex $v$ is equal to the amount of the flow outgoing from $v$, \textit{i.e.}, $\sum_{u : u \sim v} f(u, v) = 0$.

Each circulation $f$ of $G_b$ maps to a (possibly unnormalized) flip state $\ket{\phi'}$ of $G$ (whose normalized form we denote by $\ket{\phi}$), and vice versa, by the following bijection:
\begin{equation}
	\label{eq:bij}
	f(u_{\tout}, v_{\tin})= \braket{uv | \phi'}, \quad f(v_{\tin}, u_{\tout}) = -\braket{uv | \phi'}.
\end{equation}
The amplitudes of the outgoing edges of $u$ sum up to 0, thus the flow is conserved at vertex $u_{\tout}$. Similarly flow conservation holds also at vertex $u_{\tin}$.

Now consider an arbitrary starting state $\ket{\psi_0}$. To find the value of $|\alpha|^2$, we need the flip state $\ket{\phi}$ that is the closest to $\ket{\psi_0}$ among all the flip states. Alternatively, we can search for an optimal circulation in $G_b$. Next we show that a circulation that is sufficiently close to optimal can be obtained using electric networks.


\subsection{Electric Networks and Oscillations}

We examine \emph{electric networks} \cite{Doyle1984}, which are graphs where each edge is replaced by a unit resistor. Each vertex of the graph may also be either a source or a sink of some amount of current. We construct an electric network $\mathcal N_b$ from $G_b$ and $\ket{\psi_0}$ in the following fashion. The vertex set of $\mathcal N_b$ is equal to that of $G_b$. Examine the amplitude $\braket{uv|\psi_0} = \delta$ at an edge $\ket{uv}$.
\begin{enumerate}[label=(\alph*)]
	\item If $\delta = 0$, add an edge $\{u_{\tout}, v_{\tin}\}$ to $\mathcal N_b$ with a unit resistance assigned.
	\item If $\delta \neq 0$, inject $\delta$ units of current at $v_{\tin}$ and extract the same amount at $u_{\tout}$.
\end{enumerate}
Note that $\mathcal N_b$ could have multiple sources and sinks of current by the construction.

For example, let $G$ be a complete graph of three vertices $a, b, c$, which we considered in Fig.~\ref{fig:complete3}. Its double bipartite graph $G_b$ is a cycle of length 6, as shown in Fig.~\ref{fig:complete3_double_rearrange}. Say the starting state of the walk is $\ket{\psi_0} = \ket{ab}$. Then the electric network $\mathcal N_b$ is a path of length 5, as the edge $\{a_{\tout}, b_{\tin}\}$ is excluded from the electric network, as shown in Fig.~\ref{fig:complete_network}. A unit current is then injected at $b_{\tin}$ and extracted at $a_{\tout}$. 

\begin{figure}
\begin{center}
	\includegraphics{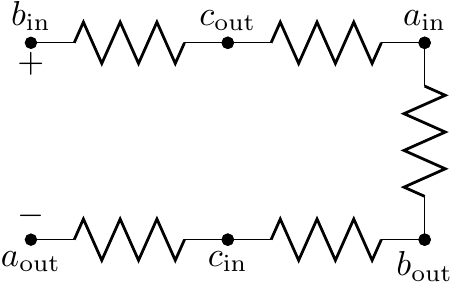}
	\caption{\label{fig:complete_network} The electric network $\mathcal N_b$ for the complete graph of three vertices $a, b, c$ with starting state $\ket{\psi_0} = \ket{ab}$.}
\end{center}
\end{figure}

Let $I_b(x,y)$ be the current flowing on the edge from $x$ to $y$ in $\mathcal N_b$. From this current, we can construct a circulation $f$ of $G_b$. Consider an edge $\{u_{\tout}, v_{\tin}\}$ in $G_b$.
\begin{enumerate}[label=(\alph*)]
	\item	If $\{u_{\tout}, v_{\tin}\} \in E(\mathcal N_b)$, set
		\[ f(u_{\tout}, v_{\tin}) = -f(v_{\tin}, u_{\tout}) = I_b(u_{\tout}, v_{\tin}). \]
	\item	Otherwise, set 
		\[ f(u_{\tout}, v_{\tin}) = -f(v_{\tin}, u_{\tout}) = \braket{uv|\psi_0}. \]
\end{enumerate}
By construction, $f$ satisfies skew symmetry. Since the net current flowing into a vertex equals the net current leaving it, $f$ also satisfies flow conservation, so it is a circulation.

By the bijection \eqref{eq:bij}, this circulation $f$ corresponds to a flip state $\ket{\phi'}$, which in general is unnormalized. Since each amplitude in $\ket{\phi'}$ is either an amplitude in $\ket{\psi_0}$ or a current in $\mathcal N_b$, we have
\[ \sum_{(u,v)} \left| \braket{uv|\phi'} \right|^2 = \sum_{(u,v)} \left| \braket{uv|\psi_0} \right|^2 + \!\!\!\!\!\!\!\!\! \sum_{\{x,y\} \in E(\mathcal N_b)} \!\!\!\!\!\!\!\!\! |I_b(x, y)|^2. \]
Since the starting state $\ket{\psi_0}$ is normalized, the first sum on the right-hand side is $1$. For the rightmost sum, note the power dissipation through a resistor is $I^2 R$ \cite{Giancoli2008}, where $I$ is the current through the resistor and $R$ its resistance. In our network, $R = 1$ since we have unit resistors, so the power dissipation is $I^2$. Then the rightmost sum is equal to the power dissipation $P(\mathcal{N}_b)$ in the network, i.e.,
\[ P(\mathcal N_b) = \!\!\!\!\!\!\!\!\! \sum_{\{x,y\} \in E(\mathcal N_b)} \!\!\!\!\!\!\!\!\! |I_b(x,y)|^2. \]
Therefore,
\[ \sum_{(u,v)} \left| \braket{uv|\phi'} \right|^2 = 1 + P(\mathcal N_b). \]

It is also possible that the current cannot flow in $\mathcal N_b$, depending on the starting state $\ket{\psi_0}$. For example, there is no flow if $G$ is the complete graph of three vertices with $\ket{\psi_0} = \frac{1}{\sqrt 2} (\ket{ab} + \ket{ba})$. In this case, we regard the power dissipation as being infinitely large.

Normalizing $\ket{\phi'}$ and calling it  $\ket{\phi}$,
\[ \ket{\phi} = \frac{1}{\sqrt{1 + P(\mathcal N_b)}} \ket{\phi'}. \]
We are interested in how close this flip state $\ket{\phi}$ is to the starting state $\ket{\psi_0}$. This is given by the inner product
\begin{align*}
	\braket{\psi_0|\phi}
		&= \frac{1}{\sqrt{1 + P(\mathcal N_b)}} \braket{\psi_0|\phi'} \\
		&= \frac{1}{\sqrt{1 + P(\mathcal N_b)}} \sum_{(u,v)} \braket{\psi_0|uv} \braket{uv|\phi'} \\
		&= \frac{1}{\sqrt{1 + P(\mathcal N_b)}} \sum_{(u,v)} \left| \braket{\psi_0|uv} \right|^2 \\
		&= \frac{1}{\sqrt{1 + P(\mathcal N_b)}}.
\end{align*}
To go from the second to third line, recall when $\braket{uv|\psi_0} = \delta \ne 0$, then $\delta$ units of current are injected at $v_{\tin}$ and extracted from $u_{\tout}$. Then $f(u_{\tout},v_{\tin}) = \delta$, and from the bijection \eqref{eq:bij}, $\braket{uv|\phi'} = \delta = \braket{uv|\psi_0}$. For the last line, $\ket{\psi_0}$ is a normalized state. Hence the contribution of the flip states to the starting state is lower bounded by
\begin{equation}
	\label{eq:alpha_bound}
	|\alpha|^2 \geq \left|\braket{\psi_0|\phi}\right|^2 = \frac{1}{1 + P(\mathcal N_b)}.
\end{equation}
This lower bound is maximized when the power dissipation is minimized. Thomson's Principle states that the values of the current determined by Kirchhoff's Circuit Laws minimize the power dissipation. Hence the current through the network $\mathcal{N}_b$ naturally yields the best bound in \eqref{eq:alpha_bound}.

Together with Theorem \ref{thm:localization-bound}, this gives the following estimate on the oscillations: after an even number of steps $2t$ or an odd number of steps $2t+1$,
\begin{align}
	&\left| \braket{ \psi_0 | U^{2t} | \psi_0 } \right|, \big| \langle \flip{\psi_0} | U^{2t+1} | \psi_0 \rangle \big| 
		\geq 2|\fcoef|^2 - 1 \notag \\
		&\quad\quad\quad \geq 2\cdot \frac{1}{1 + P(\mathcal N_b)} - 1 
		= \frac{1-P(\mathcal N_b)}{1+P(\mathcal N_b)}. \label{eq:localization-energy}
\end{align}
This leads to our first and most general result linking electric circuits and oscillations:
\begin{theorem}
	For an arbitrary starting state $\ket{\psi_0}$, low power dissipation $P(\mathcal N_b) < 1$ implies oscillations between $\ket{\psi_0}$ and $\ket{\flip{\psi_0}}$.
\end{theorem}
In general, the power dissipation of a circuit may be hard to find. But for some initial states, we can relate it to effective resistance, which is often easier to find and is well-studied for a number of graphs.


\subsection{Oscillatory Localization of Single-Edge States}

Consider the starting state $\ket{\psi_0} = \ket{ab}$, where the walker is initially localized at vertex $a$ and points towards vertex $b$. Then the current in $\mathcal N_b$ is a unit flow (flow of one unit from $b_{\tin}$ to $a_{\tout}$), and thus $P(\mathcal N_b) = R(\mathcal N_b)$ where $R(\mathcal N)$ denotes the effective resistance of $\mathcal{N}$.

A related notion is the \emph{resistance distance} $\Omega_{a,b}$ between two vertices $a$ and $b$ in a graph \cite{Klein1993}. It is the effective resistance between $a$ and $b$ in a network obtained from the graph by replacing each edge by a unit resistor. 

The only edge in which $G_b$ and $\mathcal N_b$ differ is $\{a_{\tout},b_{\tin}\}$, so we may interpret $G_b$ as an electric network where $\{a_{\tout},b_{\tin}\}$ with a unit resistance and $\mathcal N_b$ are connected in parallel. Hence low resistance distance $\Omega_{a_{\tout},b_{\tin}}$ in $G_b$ implies low $R(\mathcal N_b)$:
\begin{equation}
	\label{eq:resistance}
	\Omega_{a_{\tout},b_{\tin}} = \frac{1}{\frac{1}{1} + \frac{1}{R(\mathcal N_b)}} = 1 - \frac{1}{1+R(\mathcal N_b)},
\end{equation}
where we used the parallel resistance formula \cite{Giancoli2008}. Substituting in \eqref{eq:localization-energy}, we obtain
\begin{equation}
\left| \braket{ \psi_0 | U^{2t} | \psi_0 }\right|, \big| \langle \flip{\psi_0} | U^{2t+1} | \psi_0 \rangle \big| \geq 1-2\Omega_{a_{\tout},b_{\tin}}.
\end{equation}
Since $\ket{\flip{\psi_0}} = -\ket{ba}$, we have proven that
\begin{theorem}
	\label{thm:loc1}
	Low resistance distance $\Omega_{a_{\tout},b_{\tin}} < 1/2$ in $G_b$ implies oscillatory localization between $\ket{\psi_0} = \ket{ab}$ and $\ket{\flip{\psi_0}} = -\ket{ba}$.
\end{theorem}

Notably, if $G$ is bipartite, then $G_b$ consists of two identical copies of $G$. It follows that $\Omega_{a,b} = \Omega_{a_{\tout},b_{\tin}}$ because the current may flow from $a_{\tout}$ to $b_{\tin}$ in only one of the copies. Then low effective resistance in the original graph $G$ is sufficient to imply localization of $\ket{ab}$.


\subsection{Localization of Self-Flip States}

For a particular set of starting states, we can essentially repeat the same analysis with the original graph $G$, not its bipartite double graph $G_b$. We say that $\ket{\psi}$ is a \emph{self-flip state} if $\braket{uv|\psi} = -\braket{vu|\psi}$ for each edge $\ket{uv}$. That is, $\ket{\psi} = \ket{\flip{\psi}}$. Then a self-flip state undergoing oscillations is stationary, since it alternates between $\ket{\psi}$ and $\ket{\flip{\psi}} = \ket{\psi}$.

Suppose the initial state $\ket{\psi_0}$ is a self-flip state. We construct an electric network $\mathcal N$ with the same vertex set as $G$. Examine each \emph{undirected edge} $\{u, v\}$ once; let $\braket{uv|\psi_0} = \delta$.
\begin{enumerate}[label=(\alph*)]
	\item If $\delta = 0$, add an edge $\{u, v\}$ to $\mathcal N$ with a unit resistance assigned.
	\item If $\delta \neq 0$, inject $\delta$ units of current at $v$ and extract the same amount at $u$.
\end{enumerate}

Let $I(u,v)$ again be the current flowing on the edge from $u$ to $v$ in $\mathcal N$. Since $\braket{uv|\psi_0}=-\braket{vu|\psi_0}$ and $I(u,v) = -I(v,u)$, we can construct a circulation $f$ of $G$ similarly to the previous construction of $f$ from $I_b$. For an edge $\{u, v\}$ in $G$:
\begin{enumerate}[label=(\alph*)]
	\item	If $\{u, v\} \in E(\mathcal N)$, set
		\[ f(u, v) = -f(v, u) = I(u, v). \]
	\item	Otherwise, set 
		\[ f(u, v) = -f(v, u) = \braket{uv|\psi_0}. \]
\end{enumerate}

Let $\ket{\phi'}$ be the unnormalized flip state corresponding to $f$ according to \eqref{eq:bij}. Again, since the amplitude on each edge $\ket{uv}$ in $\ket{\phi'}$ is either $I(u,v)$ or $\braket{uv|\psi_0}$, we have
\[ \sum_{(u,v)} \left| \braket{uv|\phi'} \right|^2 = \sum_{(u,v)} \left| \braket{uv|\psi_0} \right|^2 + 2 \!\!\!\!\!\!\!\!\! \sum_{\{x,y\} \in E(\mathcal N)} \!\!\!\!\!\!\!\!\! |I(x, y)|^2. \]
The last sum has a factor of 2 as each edge $\{u,v\}$ from $\mathcal N$ sets the amplitudes of both $\ket{uv}$ and $\ket{vu}$ by the construction of $f$. Thus the expression is equal to $1 + 2P(\mathcal{N})$. By the same procedure as before, we use this to find the normalized state $\ket{\phi}$, the overlap $\braket{\psi_0|\phi}$, and a lower bound on $|\alpha|^2$, yielding a lower bound on localization:
\begin{equation}
	\label{eq:self-flip-energy}
	\left| \braket{ \psi_0 | U^{2t} | \psi_0 }\right|, \big| \langle \flip{\psi_0} | U^{2t+1} | \psi_0 \rangle \big| \geq \frac{1-2P(\mathcal N)}{1+2P(\mathcal N)}.
\end{equation}
This is similar to \eqref{eq:localization-energy}, and it yields the next result:

\begin{theorem}
	For a self-flip starting state $\ket{\psi_0}$, low power dissipation $P(\mathcal N) < 1/2$ implies that $\ket{\psi_0}$ is stationary.
\end{theorem}

We can relate this to effective resistance for the particular starting state $\ket{\psi_0} = (\ket{ab} - \ket{ba})/\sqrt{2}$, where the particle is localized at vertices $a$ and $b$. By constructing $\mathcal N$ as described before, we obtain a current of $1/\sqrt 2$ units from $a$ to $b$. Then the power dissipation is equal to this current squared times the effective resistance, so $2P(\mathcal N) = R(\mathcal N)$. Then similarly to \eqref{eq:resistance}, we deduce that
\[ \Omega_{a,b} = \frac{1}{\frac{1}{1} + \frac{1}{R(\mathcal N)}} = 1 - \frac{1}{1+R(\mathcal N)}. \]
Therefore, substituting this in \eqref{eq:self-flip-energy}, we get
\[ \left| \braket{ \psi_0 | U^{2t} | \psi_0 }\right|, \big| \langle \flip{\psi_0} | U^{2t+1} | \psi_0 \rangle \big| \geq 1 - 2\Omega_{a,b}. \]
This proves another new result:
\begin{theorem}
	\label{thm:loc2}
	Low resistance distance $\Omega_{a,b} < 1/2$ in $G$ implies localization of the starting state $\ket{\psi_0} = (\ket{ab} - \ket{ba})/\sqrt{2}$.
\end{theorem}


\subsection{High Connectivity}

Resistance distance is known to be closely related to edge-connectivity.
Two vertices $s$ and $t$ are said to be $k$-edge-connected if there exists $k$ edge-disjoint paths from $s$ to $t$.

A tight relation was proved in \cite{Ahn2013}: if $s$ and $t$ are $k$-edge-connected, then
\begin{equation}
	\label{eq:connectivity-bound}
	\Omega_{s,t} = O\left( \frac{N^{2/3}}{k} \right).
\end{equation}

If one also looks at the lengths of the paths, then a stronger statement is true. Let $s$ and $t$ be connected by $k$ edge-disjoint paths with lengths $\ell_1, \ell_2, \ldots, \ell_k$. The resistance distance between $s$ and $t$ in the subgraph induced by these paths may only be larger than in the original graph. In this subgraph, the paths correspond to $k$ resistors connected in parallel with resistances equal to $\ell_1, \ell_2, \ldots, \ell_k$. Then we have the following upper bound:
\begin{equation}
	\label{eq:paths}
	\Omega_{s,t} \leq \frac{1}{\sum_{i=1}^k \frac{1}{\ell_i}}.
\end{equation}

In the context of this paper, we arrive at the following conclusion by Theorems~\ref{thm:loc1} and \ref{thm:loc2}, that high edge-connectivity implies low resistance distance:
\begin{theorem}
	High edge-connectivity between $a_{\tout}$ and $b_{\tin}$ in $G_b$ implies oscillatory localization between $\ket{ab}$ and $-\ket{ba}$. On the other hand, high edge-connectivity between $a$ and $b$ in $G$ implies that $\ket{\psi_0} = (\ket{ab} - \ket{ba})/\sqrt{2}$ is stationary.
\end{theorem}

In particular, $k = \omega(N^{2/3})$ means that $\Omega_{s,t} = o(1)$ by \eqref{eq:connectivity-bound} and therefore implies localization. For instance, edge-connectivity between any two vertices in the complete graph is high: $\Theta(N)$ and then $\Omega_{s,t} = O\left(1/N^{1/3}\right)$. On the other hand, edge-connectivity between two vertices $s$ and $t$ connected by an edge in a $d$-dimensional hypercube is only $\Theta(d) = \Theta(\log N)$, so \eqref{eq:connectivity-bound} is not enough. There are $\Theta(d)$ edge-disjoint paths of length 3, however, between $s$ and $t$; thus by \eqref{eq:paths}, the resistance is small: $\Omega_{s,t} = O\left(3/\log N\right)$. Therefore, edge-connectivity is another useful measure of graphs that implies (oscillatory) localization of quantum walks.


\section{\label{sec:examples} Examples}

A large variety of regular graphs have low resistance distance. For instance, the resistance distance between any two vertices $u$ and $v$ connected by an edge in a $d$-regular edge-transitive graph is given by \cite{Foster1949, Klein1993}:
\begin{equation}
	\label{eq:omega_ab}
	\Omega_{u,v} = \frac{N-1}{dN/2} \approx \frac{2}{d}.
\end{equation}
Similarly for the bipartite double graph,
\begin{equation}
	\label{eq:omega_ab_double}
	\Omega_{u_{\tout},v_{\tin}} = \frac{2N-1}{dN} \approx \frac{2}{d}.
\end{equation}
Thus in any case, the resistance distance is small provided the degree $d$ is high. Then from Theorems~\ref{thm:loc1} and \ref{thm:loc2}, $\ket{ab}$ oscillates with $-\ket{ba}$, while $(\ket{ab}-\ket{ba})/\sqrt{2}$ is localized, for edge-transitive graphs including complete graphs, complete bipartite graphs, hypercubes, and arbitrary-dimensional square lattices with degree greater than four. In addition, both of these states correspond to a single edge in the electric network. Then the current, which minimizes the energy dissipation, exactly corresponds to the flip state closest to $\ket{\psi_0}$, which implies equality in \eqref{eq:alpha_bound}. Then $|\alpha|^2$ can be found exactly by substituting \eqref{eq:omega_ab} or \eqref{eq:omega_ab_double} into \eqref{eq:resistance} and then into \eqref{eq:alpha_bound}, yielding $|\alpha|^2 = (dN - 2N + 2)/dN$ and $(dN - 2N + 1)/dN$, respectively.

Using this result, let us revisit our initial example in Fig.~\ref{fig:complete} of a quantum walk on the complete graph with starting state $\ket{\psi_0} = \ket{ab}$. Then the degree is $d = N-1$, and the probability overlap of the starting state with a flip state is
\[ |\fcoef|^2 =  \frac{N(N-1)-2N+1}{N(N-1)} =1 - \frac{1}{N-1} - \frac 1 N. \]
The value of $|\ucoef|^2$ can also be calculated explicitly as
\[ |\ucoef|^2 = \left|\braket{ab | \ustate_V}\right|^2 = \left|\frac{1}{\sqrt{dN}}\right|^2 = \frac{1}{N(N-1)}. \]
Using these precise values of $|\fcoef|^2$ and $|\ucoef|^2$ in Theorem \ref{thm:localization-bound}, the amplitude of the state being in its initial state $\ket{ab}$ at even timesteps is
\begin{align*}
	\left| \braket{ ab | U^{2t} | ab } \right|
	&\geq 2\left(|\fcoef|^2 + |\ucoef|^2\right)-1 \\
	&= 2\left(1 - \frac{1}{N-1} - \frac 1 N + \frac{1}{N(N-1)} \right) - 1\\
	&= 1 - \frac 4 N.
\end{align*}
On the other hand, the amplitude of being in its flipped version $-\ket{ba}$ at odd timesteps is
\begin{align*}
	\left| \braket{ ba | U^{2t+1} | ab } \right| 
	&\geq 2\max\left(1 - \frac{1}{N-1} - \frac 1 N,\frac{1}{N(N-1)}\right)-1 \\
	&= 2\left(1 - \frac{1}{N-1} - \frac 1 N\right)-1 \\
	&= 1 - \frac{2}{N-1} - \frac 2 N.
\end{align*}
Let us compare these bounds to the exact result in Section~\ref{sec:complete}. There, we found in \eqref{eq:complete_ab} that at even steps, the system was in $\ket{ab}$ with amplitude
\[ \frac{N-2}{N} + \frac{2}{N} \cos(\theta t) = 1 - \frac{2}{N} + \frac{2}{N} \cos(\theta t) \ge 1 - \frac{4}{N}, \]
so the bound we just derived from Theorem~\ref{thm:localization-bound} is tight. From Section~\ref{sec:complete}, we found in \eqref{eq:complete_ba} that at odd steps that the system was in $- | ba \rangle$ with amplitude
\[ \frac{N-3}{N-1} = 1 - \frac{2}{N-1}. \]
So the bound we just derived is close to being tight.

Expander graphs \cite{Hoory2006} are generally \emph{not} edge-transitive, and they have significant applications in communication networks, representations of finite graphs, and error correcting codes. Expander graphs of degree $d$ have resistance distance $\Theta(1/d)$ \cite{Chandra1996}, hence expander graphs of degree $d = \omega(1)$ exhibit localization of the initial state $(\ket{ab}-\ket{ba})/\sqrt{2}$.

A general graph, which may be irregular, also has low resistance distance $\Omega_{a,b} \leq \frac 4 d \leq \frac 8 N$ when its minimum degree $d \geq \lfloor\frac N 2\rfloor$ is high \cite{Chandra1996}. So $(\ket{ab}-\ket{ba})/\sqrt{2}$ is also localized for these graphs.


\section{\label{sec:conclusion} Conclusion}

We have introduced a new type of localization where a quantum walk alternates between two states. Such oscillation is exactly characterized by only two types of states, flip states and uniform states, both of which are 1-eigenvectors of $U^2$. By projecting the initial localized state $\ket{\psi_0}$ of the system onto these states, their respective amplitudes $\fcoef$ and $\ucoef$ give bounds on whether oscillatory localization occurs.

While finding $\beta$ exactly is trivial, simply bounding $\alpha$ requires more advanced analysis, which we provided using a bijection between flip states and current in electric networks. Using this, we proved a general result that oscillations on a graph $G$ occur when the power dissipation of an electric network defined on the bipartite double graph $G_b$ is low. For the initial state $\ket{ab}$, low effective resistance on $G_b$ can be used instead, while for $(\ket{ab}-\ket{ba})/\sqrt{2}$, low effective resistance on the original graph $G$ suffices. We can also use high connectivity instead of low effective resistance. Thus quantities in classical electric circuits imply a quantum behavior on graphs, namely oscillatory localization of quantum walks.

Further research includes relating power dissipation, effective resistance, and connectivity to other properties to determine what other starting states and graphs exhibit oscillatory localization. For example, Menger's theorem \cite{Menger1927,Halin1974} states that the maximum number of edge-disjoint paths between $s$ and $t$ is equal to the minimum $s$-$t$ edge-cut. Resistance is also related to commute time and cover time \cite{Chandra1996}. Oscillatory localization can also be investigated for quantum walks on different graphs or with different initial states, or for different definitions of quantum walks. This opens up the exploration of oscillations where the particle does not exactly return to its initial state, but has a global, unobservable phase. Finally, as our exact oscillatory localization is an example of a quantum walk periodic in two steps, another open area is investigating quantum walks with larger periods.


\begin{acknowledgments}
	This work was supported by the European Union Seventh Framework Programme (FP7/2007-2013) under the QALGO (Grant Agreement No.~600700) project and the RAQUEL (Grant Agreement No.~323970) project, the ERC Advanced Grant MQC, and the Latvian State Research Programme NeXIT project No.~1.
\end{acknowledgments}


\bibliography{refs}

\end{document}